\documentclass[journal]{IEEEtran}
\usepackage{graphicx} %
\usepackage{algorithm}
\usepackage{algpseudocode}
\usepackage{comment}
\usepackage{amsmath}

\usepackage{amsthm}
\usepackage{thmtools,thm-restate}      
\usepackage{amssymb}
\usepackage{tikz}
\usepackage{caption}
\usepackage{subcaption}
\usepackage{float}
\usepackage{packages/listingsutf8}
\usepackage{float}

\usepackage{hyperref}                  
\usepackage[capitalise]{cleveref} 
\usepackage{enumitem}

\usepackage{thmtools}

\usepackage{cancel}

\newtheorem{definition}{Definition}

\newtheorem{theorem}{Theorem}
\newtheorem{lemma}{Lemma}

\newcommand{\vect}[1]{\mathbf{#1}}

\title{Physical Time-Lock Puzzles}
\author{Niloufar~Sayadi,
        Chenglu~Jin,
        Phuong~Ha~Nguyen,
        Zheng~Yang,
        and~Marten~van~Dijk%
\thanks{N. Sayadi is with Centrum Wiskunde \& Informatica and Vrije Universiteit Amsterdam, the Netherlands (e-mail: niloufar.sayadi@cwi.nl).}%
\thanks{C. Jin is with Centrum Wiskunde \& Informatica, Amsterdam, the Netherlands (e-mail: chenglu.jin@cwi.nl).}
\thanks{P. H. Nguyen is with eBay, San Jose, CA, the USA (e-mail: phuongha.ntu@gmail.com).}%
\thanks{Z. Yang is with Southwest University, Chongqing, China (e-mail: youngzheng@swu.edu.cn).}
\thanks{M. van Dijk is with Centrum Wiskunde \& Informatica, Vrije Universiteit Amsterdam, the Netherlands, and the University of Connecticut, Storrs, CT, the USA (e-mail: marten.van.dijk@cwi.nl).}%
}

\begin{document}
\maketitle

\begin{abstract}
Traditional time-lock puzzles enforce delayed access to encrypted secrets by relying on inherently sequential computational steps. However, since software-based constructions impose no fundamental bound on the physical execution speed of individual steps, they remain vulnerable to hardware acceleration and improved implementations. In practice, this renders existing schemes ``step-lock'' rather than true ``time-lock'' puzzles, making long-term delay guarantees highly speculative against decades of unpredictable hardware advancement. 

To address this fundamental limitation, we introduce the \emph{Physical Time-Lock Puzzle} (P-TLP), a new paradigm that anchors solving delay directly to the intrinsic, hardware-bounded evaluation latency of silicon hardware. By leveraging noisy Physical Unclonable Functions (PUFs) as non-parallelizable delay oracles, P-TLPs establish resistance against both parallel computation and algorithmic acceleration under the random oracle model. To tolerate PUF evaluation noise while maintaining a tightly concentrated and predictable solving-delay window, we propose a composite puzzle architecture that combines multiple independent basic puzzles. 

We formally prove the optimality of a sequential greedy solving strategy and derive tightly concentrated solving-delay windows using Hoeffding and Berry-Esseen bounds. We additionally prove early-solve hardness, establishing that adversaries with substantial classical computation budgets cannot bypass the physically-enforced PUF evaluation bottleneck. We validate our theoretical analysis with an FPGA prototype built around a configurable Ring Oscillator PUF, confirming that our analytical delay predictions tightly match empirical measurements. Our results demonstrate that P-TLPs can deliver highly predictable, long-term delay guarantees, enabling practical deployment in high-value applications such as digital legacy management.

\end{abstract}


\section{Introduction}
\label{sec:intro} 
Informally, a time-lock puzzle is a cryptographic mechanism for effectively sending messages into the future~\cite{RivestShamirWagner1996,MahmoodyMoranVadhan2011,BitanskyGoldwasserJainEtAl2016}. Specifically, a message is encrypted such that it can only be decrypted after a strictly specified amount of time has elapsed, but never beforehand. In this context, the encrypted ciphertext acts as the puzzle, and the original message is the solution successfully retrieved after a defined time parameter $T$. From a security perspective, no adversary should be able to solve the puzzle in time significantly less than $T$, even when natively equipped with massively parallel computing resources.

Early constructions of time-lock puzzles are purely software-based, relying exclusively on inherently sequential computations. The pioneering construction by Rivest, Shamir, and Wagner~\cite{RivestShamirWagner1996} is based on the serial nature of repeated squaring and the hardness of factoring, while recent schemes have explored randomized encodings~\cite{BitanskyGoldwasserJainEtAl2016} and lattice problems~\cite{agrawalr2024time}. In all of these existing software constructions, the solving algorithms strictly demand an \textit{inherently sequential} computational execution path.

Most software-based time-lock puzzles bound computational steps instead of execution time, acting as step-lock puzzles sensitive to hardware speed. Developing puzzles independent of raw computation remains an open challenge. Rivest’s 1999 puzzle~\cite{Rivest1999} illustrates this difficulty, as its 35-year estimate relied on speculative Moore’s Law projections to account for future hardware acceleration.

In practice, this speculative delay guarantee often appeared overly optimistic. The puzzle, intended to be secure until 2034, was solved 15 years early in 2019. A software implementation utilizing the GNU Multiple Precision Arithmetic Library on a commercial desktop CPU successfully extracted the solution after 3 years and 3 months of continuous execution~\cite{Rivest2019}. Even more remarkably, an independent group leveraged customized FPGA hardware to accelerate the underlying operations, solving the identical puzzle in a mere 2 months~\cite{Rivest2019}. This vividly demonstrated the significant inaccuracies and vulnerabilities inherent to purely software-based delay guarantees against customized hardware acceleration.

To tackle these fundamental vulnerabilities, we ask a core question: \textbf{can we cryptographically bind computation time directly to a specific physical hardware, ensuring that only this exact piece of hardware can be used to solve the puzzle?} If so, the required solving time immediately becomes independent of advances in general-purpose algorithmic computing and external implementation variance. Consequently, constructing a highly predictable, long-term time-lock puzzle becomes practically feasible.

Silicon Physical Unclonable Functions (PUFs)~\cite{Gassend2002} are hardware primitives that leverage uncontrollable and unclonable manufacturing variations to produce unique challenge-response behavior. For physical time-lock puzzles, we specifically exploit their intrinsic evaluation latency $\tau_P$. Because this propagation delay is rooted in physical electronics, no adversary can evaluate a PUF faster than the hardware-defined latency $\tau_P$.

Essentially, we utilize the PUF as a non-parallelizable random oracle within our puzzle construction. During generation, a generator securely queries the PUF with a randomly chosen challenge; to subsequently unlock the puzzle, a solver must perform an exhaustive search across the entire defined challenge space to locate that exact query. Because evaluating a massive search space sequentially is inherently more time-consuming than executing a single generation query, this asymmetrical workload naturally establishes a structural efficiency gap between the puzzle generator and the solver. A seminal result by Mahmoody et al.~\cite{MahmoodyMoranVadhan2011} established a formal negative result in the standard random-oracle model, strictly ruling out time-lock puzzles that require more parallel time to solve than the total work required to generate the puzzle. However, by substituting the parallelizable random oracle with a physical silicon PUF, we formally anchor the adversary to strict, \emph{non-parallelizable} physical hardware delay. This enables us to successfully bypass the theoretical impossibility result of random oracle-based time-lock puzzles.

Unlike traditional software-based TLPs, our P-TLPs require physical PUF access for both generation and solving. While public solvability is a common feature of software TLPs, our physical construction intentionally shifts the paradigm by binding the delay strictly to a specific physical token. This novel trade-off successfully addresses hardware acceleration vulnerabilities while inevitably losing compatibility with certain traditional TLP applications that require purely public puzzles.  Consequently, P-TLPs are uniquely suited for high-value applications requiring tangible physical anchoring:

\begin{itemize}[noitemsep, topsep=0pt]
\item \textbf{Digital Legacy Management}: By encrypting virtual assets within a hardware token, a creator ensures heirs cannot access the legacy prematurely. Unlike traditional software-based puzzles, the physical token prevents remote copying and stealth solving.
\item \textbf{Delayed-Access Hardware Wallets}: P-TLPs can secure cryptocurrency keys for predetermined lockdown periods, binding access to a physical PUF to mitigate permanent asset loss from forgotten passwords or sudden death.
\item \textbf{Strict Key-Escrow Schemes}: P-TLPs establish physically regulated escrow procedures, structurally guaranteeing that restricted message keys remain mathematically inaccessible to regulators until an explicit time period has elapsed.

\end{itemize}
\paragraph{Our Contributions.}
In this paper, we propose a \emph{Physical Time-Lock Puzzle (P-TLP)} that leverages a noisy silicon PUF as a non-parallelizable hardware delay primitive. Compared to traditional software-based time-lock puzzles, our approach binds the solving effort to the intrinsic evaluation latency of a specific hardware device. Our main contributions are\footnote{Importantly, designing PUFs immune to advanced machine learning modeling attacks remains an open problem. Our work establishes the foundational P-TLP architecture; the FPGA implementation presented in this paper is strictly a timing and error-correction proof-of-concept, not a claim of modeling resilience. As designing secure, lightweight, strong PUFs remains a highly active area of hardware research, our framework can seamlessly integrate future state-of-the-art PUF designs.}: 
\begin{itemize}
\item We introduce the formal framework for P-TLP operating over noisy PUF oracles. We construct a basic puzzle and rigorously prove the optimality of the proposed solving algorithm.
\item We present an advanced, composite P-TLP construction that can better tolerate PUF noise and restricts variance without compromising the intrinsic efficiency gap.
\item We derive tight analytical windows for honest solving delays utilizing concentrated Hoeffding bounds. Furthermore, we establish a guarantee of early-solve hardness, ensuring that adversaries granted substantial classical oracle testing budgets cannot bypass the physically-enforced PUF evaluation limits.
\item We implement the proposed P-TLP practically on an FPGA using a configurable Ring Oscillator PUF, successfully validating our tight theoretical delay derivations. Finally, we demonstrate the unique advantages of anchoring time-locks to physical hardware for high-value applications such as digital legacy management.

\end{itemize}

\paragraph{Organization.} Section~\ref{sec:notation_background} reviews relevant background and related work. Section~\ref{sec:Silicon_tlp} introduces the definitional framework of P-TLP. Section~\ref{sec:basicPUZ} presents a basic P-TLP construction with its analysis, while Section~\ref{sec:composition_puz} details our composite architecture. We establish formal security proofs in Section~\ref{sec:secproof} and provide experimental FPGA validations in Section~\ref{sec:Experimets}. We discuss more details about the potential applications of P-TLP in the digital legacy management in Section~\ref{sec:applications}. Finally, we conclude our paper and discuss future work in Section~\ref{sec:conclusion}.

\section{Background and Related Work}\label{sec:notation_background}

\subsection{Physically Unclonable Functions}
\label{subsec:PUF_intro}
Intuitively, a silicon Physically Unclonable Function (PUF) embedded within a microchip acts as a hardware fingerprint possessing three critical properties: (1) it evaluates a given challenge (input) $\vect{c}$ to generate a random-looking response (output) $\vect{r}$ by leveraging intrinsic manufacturing process variations; (2) it is physically unclonable; and (3) it is mathematically unclonable. Mathematical unclonability ensures that the physical function cannot be efficiently modeled or predicted in software, even when an adversary is given a polynomial number of challenge-response pairs.

Over the past decade, a variety of strong PUF constructions have been proposed, including Arbiter PUFs~\cite{Gassend2002}, Ring Oscillator PUFs~\cite{Suh2007}, XOR APUFs~\cite{Suh2007}, Lightweight Secure PUFs~\cite{Majzoobi2008a}, Composite PUFs~\cite{Sahoo2014}, LPN-based PUFs~\cite{HerderRenDijkEtAl2016,jin2017fpga}, Interpose PUFs~\cite{NguyenSJMRD19}, and LP PUFs~\cite{wisiol2022towards}.

Despite these advances, many silicon PUF constructions remain vulnerable to physical cloning~\cite{Helfmeier2013,Tajik2014} and advanced cryptanalysis~\cite{Nguyen2015,Sahoo2015,kraleva2022cryptanalysis}. Machine-learning-based modeling attacks have proven especially potent~\cite{Lim2004,Ruhrmair2010,Ruhrmair2013,Tobisch2015,Becker2015,mishra2024calypso,sayadi2025breaking}.

Importantly, our P-TLP framework is agnostic to the underlying PUF architecture, requiring only that the implemented strong PUF maintains mathematical security against modeling attacks given the specified oracle access. As designing secure, lightweight strong PUFs remains a highly active area of hardware research, our construction can be seamlessly instantiated with future state-of-the-art PUF designs. Furthermore, the physical evaluation time of a PUF constitutes a fundamental hardware measurement that cannot be bypassed or parallelized. This unique physical bound makes PUFs exceptionally well-suited for enforcing hardware-bound solving delays in physical time-lock puzzles, assuming the adversary cannot physically or mathematically clone the specific PUF.

\subsection{Related Work}
\label{subsec:related}
Time-lock puzzles were first introduced by Rivest, Shamir, and Wagner~\cite{RivestShamirWagner1996} as a cryptographic mechanism for enforcing delayed access to encrypted information. Subsequent work by Mahmoody et al.~\cite{MahmoodyMoranVadhan2011} provided a formal treatment of time-lock puzzles and established fundamental limitations on constructing puzzles with a work gap between the generator and a parallel adversary in the random-oracle model. Later, Bitansky et al.~\cite{BitanskyGoldwasserJainEtAl2016} introduced a general standard-model framework for time-lock puzzles based on randomized encodings, demonstrating how inherently sequential computation can be enforced without relying on number-theoretic assumptions.

At a high level, most time-lock puzzle constructions follow a common pattern: puzzle generation is fast, while puzzle solving requires carrying out a sequence of computational work that cannot be parallelized. The delay is enforced by lower-bounding the number of computational steps that must be executed by any solver, even adversaries equipped with substantial parallel resources. As a result, the delay guarantees in these constructions are expressed entirely in terms of computational steps rather than elapsed physical time.

A representative example is the classical RSW construction~\cite{RivestShamirWagner1996}, which is rooted in the serial nature of modular exponentiation and the hardness of integer factoring. Specifically, to solve the puzzle, an honest solver must compute the value $b=a^{2^t} \pmod n$, where $a$ is a random integer in $\mathbb{Z}_n^*$, $n=p\times q$ is the product of two randomly generated prime numbers, and $t$ defines the time parameter for the puzzle. Without knowledge of the prime factorization of $n$, the solver is forced to perform exactly $t$ sequential squarings to compute $b$, imposing a strict computational delay. In contrast, the puzzle generator can compute $b$ highly efficiently by utilizing their knowledge of the secret primes $p$ and $q$. First, the generator computes the reduced exponent $e = 2^t \pmod{\phi(n)}$, where $\phi(n)=(p-1)(q-1)$. The value $b$ is then directly computed via a single modular exponentiation $b = a^e \pmod n$. Any attempt by an adversary to shortcut this sequential squaring process without exponentiating step-by-step computationally reduces to the well-known integer factorization problem.

In recent years, the study of time-lock puzzles has expanded significantly to support advanced cryptographic features and multi-instance environments. For example, homomorphic time-lock puzzles computationally enable operations on time-locked data without unlocking it early~\cite{malavolta2019homomorphic}, while non-malleable constraints secure puzzles against active tampering and related-puzzle attacks in strictly composable protocols~\cite{freitag2021non}. To handle multiple instances efficiently, other extensions have introduced chained structures for sequential multi-instance solving with public verifiability~\cite{abadi2021multi,abadi2025scalable}, as well as batchable architectures capable of amortizing the cost of generating or solving many puzzles simultaneously~\cite{srinivasan2023transparent,dujmovic2024time}. Closely related to this theoretical ecosystem are Verifiable Delay Functions (VDFs)~\cite{boneh2018verifiable}, a primitive utilizing similar sequential hardness assumptions to strictly certify elapsed time with polylogarithmic verification, rather than directly hiding a secret message. While all of these extensions drastically improve cryptographic functionality, scalability, and verifiability, they fundamentally continue to enforce delay strictly through sequential computational effort, leaving them inevitably sensitive to raw hardware acceleration and parallel processing speeds.

A fundamentally different approach to enforcing cryptographic delay is taken in  CaSCaDE~\cite{baum2024cascade}, which leverages physical communication latency rather than local computation. By exploiting the finite speed of light across satellite communication networks, the construction imposes a strict physical lower bound on message propagation time. Through carefully orchestrated interactive protocols among geographically distributed nodes, CaSCaDE ensures that specific protocol steps cannot be executed faster than the underlying network latency. This allows for the creation of publicly verifiable time-lock puzzles and verifiable delay functions whose security is anchored in physical constraints rather than computational hardness assumptions. However, the exact delay enforced by CaSCaDE is inherently tied to a complex distributed infrastructure, which severely restricts its practical deployability for standalone or lightweight applications.

\section{Physical Time-Lock Puzzles}\label{sec:Silicon_tlp}

\begin{definition}[PUFs - Physical Unclonable Functions~\cite{jin2022programmable}{}]\label{def:puf}
A PUF ${\mathsf{P}}$ is a physical system that can be stimulated with challenges $c_i$ drawn from a challenge set $C_{\mathsf{P}} =\{ 0,1\}^\lambda$, upon which it reacts by producing corresponding responses $r_i$ from a response set $R_{\mathsf{P}} \subseteq \{0,1 \}^m$. Each response $r_i$ depends not only on the applied challenge but also on intrinsic manufacturing variations within ${\mathsf{P}}$ that effectively render it physically unclonable with existing technology. The tuples $(c_i,r_i)$ are defined as the challenge-response pairs (CRPs) of ${\mathsf{P}}$. We generally refer to $\lambda$ as the security parameter of the PUF.
\end{definition}

\paragraph{Noisy PUF Oracle.}
Each PUF inherently defines a challenge–response mapping. However, due to natural environmental noise, repeated queries on the exact same challenge may yield slightly different responses. Definition~\ref{def:puf-indist} formalizes this physical noisy behavior by abstracting the PUF as an idealized oracle, capturing the noisy response evaluation via a per-bit error probability $e_p$.

\begin{definition}[PUF-Oracle]\label{def:puf-indist}
Let $\mathcal{O}$ be an oracle with a challenge--response (CRP) list. For each challenge $c \in \{0,1\}^\lambda$:

\begin{equation*}
\mathcal{O}(c) = 
\begin{cases} 
r \oplus \varepsilon_r, & \begin{aligned}[t]
  &\text{if } \exists(c,r) \in \text{List}, \\
  &\varepsilon_r \leftarrow \text{Bernoulli}(e_p)^m,
\end{aligned} \\[2ex]
r, & \begin{aligned}[t]
  &\text{if } (c, \cdot) \notin \text{List}, \\
  &r \leftarrow \{0,1\}^m, \\
  &\text{List} \leftarrow \text{List} \cup \{(c,r)\}.
\end{aligned}
\end{cases}
\end{equation*}

Here, $\text{Bernoulli}(e_p)^m$ determines a binary noise vector $\varepsilon_r \in \{0,1\}^m$ where each entry is i.i.d.\ and follows a Bernoulli distribution with probability parameter\footnote{In practice, $e_p$ may depend on both the challenge $c$ and its specific position within the response. For the purposes of this work, we treat $e_p$ as the expected value of its underlying distribution. Using an expected value suffices for our failure analysis without overcomplicating the noise modeling; however, we explicitly acknowledge that this simplified model makes PUF indistinguishability an idealized assumption in our formal proofs. Note that the security analysis in Section~\ref{sec:secproof} assumes the worst-case scenario for the defenders with $e_p = 0$, meaning the attackers can query the PUF without any noise.} $e_p$, and $r$ is uniformly random over $\{0,1\}^m$.
\end{definition}
While recent work~\cite{van2023theoretical} provides a precise definition of PUF-based random oracles, it incorporates error correction into its construction. In contrast, we do not attempt to correct PUF responses, so we define the PUF oracle differently.

\paragraph{PUF–Oracle Indistinguishability.}
Beyond being inherently hard to invert, a secure PUF should be computationally indistinguishable from an idealized PUF oracle as established in Definition~\ref{def:puf-indist}. We formalize this property below.

Let $\mathsf{Exp}_{\mathcal{D}}^{\mathsf{P}}(\lambda)$ and $\mathsf{Exp}_{\mathcal{D}}^{\mathcal{O}}(\lambda)$ denote the experiments where a PPT distinguisher $\mathcal{D}$ adaptively queries either a physical PUF $\mathsf{P}$ or a PUF oracle $\mathcal{O}$, respectively, and outputs a bit $b$ indicating whether $\mathcal{D}$ correctly identified the source in the experiments.

We say that PUF $\mathsf{P}$ is \emph{oracle-indistinguishable} if, for all PPT distinguishers~$\mathcal{D}$, there exists a negligible function~$\mathrm{negl}(\cdot)$ such that the advantage of distinguishing $\mathsf{P}$ from $\mathcal{O}$ is negligible in its security parameter $\lambda$, formalized as:

\[
\begin{split}
\mathrm{Adv}^{\mathrm{PUF\text{-}ind}(\mathsf{P}, \mathcal{O})}_{\mathcal{D}}(\lambda) 
&= \bigl\lvert \Pr[\mathsf{Exp}^{\mathsf{P}}_{\mathcal{D}} (\lambda)=1] -\Pr[\mathsf{Exp}^{\mathcal{O}}_{\mathcal{D}} (\lambda) = 1] \bigr\rvert \\
&\le \mathrm{negl}(\lambda).
\end{split}
\]

Defining this indistinguishability allows us to rigorously analyze our theoretical constructions in the \emph{PUF oracle model} rather than directly analyzing a concrete physical PUF. Hence, throughout the remainder of this paper, we perform all theoretical analyses by modeling the PUF as an abstract oracle $\mathcal{O}$, consistent with Definition~\ref{def:puf-indist}. %
The security parameter $\lambda$ strictly governs the indistinguishability between the physical PUF and the ideal oracle, while also defining the size of the underlying challenge space. Importantly, this definition also means that the physical PUF remains secure against any modeling attacks given bounded oracle access. Furthermore, because our puzzle framework relies on comparing exactly two PUF evaluations for a given target challenge (one during generation, one during solving), it is mathematically symmetrical to treat the first query as noise-free and attribute the aggregated variance entirely to the second query as modeled in Definition~\ref{def:puf-indist}.

\paragraph{PUF-based Time-Lock Puzzles.}
A PUF-based time-lock puzzle is constructed by successfully coupling puzzle generation with puzzle solving, both of which require active queries to the PUF $\mathsf{P}$. The puzzle generation algorithm ($\mathsf{PUZ.Gen}$) produces a secret round-key and associated helper data; conversely, the solving algorithm ($\mathsf{PUZ.Solve}$) subsequently attempts to recover the round-key given only the helper data.

\begin{definition}[PUF-based Time-Lock Puzzle Scheme]
A PUF-based time-lock puzzle $\mathsf{PUZ}$ relative to a PUF $\mathsf{P}$ is a pair of probabilistic polynomial-time (PPT) algorithms $\mathsf{PUZ} = (\mathsf{PUZ.Gen}, \mathsf{PUZ.Solve})$ strictly requiring oracle access to $\mathsf{P}$. These algorithms take an optional unique label $\ell$ (where $\ell = \mathrm{nil}$ if strict independence from other puzzle instances is not required) and operate as follows:
\begin{itemize}
    \item $(k, h; g) \gets \mathsf{PUZ.Gen}(\ell)$: Outputs a round-key $k$, helper data $h$, and the generation cost $g$ (the total number of PUF invocations).
    \item $(\hat{k}; s) \gets \mathsf{PUZ.Solve}(\ell, h)$: Outputs a candidate round-key $\hat{k}$ (or $\perp$ if no valid key is successfully recovered), and the solving cost $s$ (the total number of PUF invocations).
\end{itemize}
\end{definition}

\paragraph{Induced Random Variables.}
Executing $\mathsf{PUZ}(\ell)$ produces a joint outcome $(K,H,G,\hat{K},S)$:
the generated round key $K$, helper data $H$, generation cost $G$, recovered round key $\hat{K}$, and solving cost $S$. A concrete execution of $\mathsf{PUZ}(\ell)$ yields a realization $(k,h,g,\hat{k},s)$ drawn directly from this joint distribution.

Notably, the joint distribution $(K,H,G,\hat{K},S)$ is assumed to be identical and independent of the input label $\ell$. This i.i.d. assumption enables the analysis of total solving efforts across $n$ independent puzzle instances as a sum of independent random variables.
\begin{definition} [Induced Random Variables]\label{def:joint}
Let $\mathsf{PUZ}=(\mathsf{PUZ.Gen},\mathsf{PUZ.Solve})$ be a PUF time-lock puzzle scheme.
We say the jointly distributed random variables $(K,H,G,\hat{K},S)$ \emph{correspond to} $\mathsf{PUZ}$ if, for all input labels $\ell$, their joint distribution is defined by
\begin{multline*}
\Pr[(K,H,G,\hat{K},S) = (k,h,g,\hat{k},s)] \\
= \Pr[(k,h,g,\hat{k},s)\leftarrow \mathsf{PUZ}(\ell)].
\end{multline*}
\end{definition} 

Definitions for P-TLP properties are provided alongside their analysis in Section~\ref{sec:basicPUZ}.

\section{Basic Puzzle \texorpdfstring{$\mathsf{PUZ}$}{PUZ}}\label{sec:basicPUZ}
\subsection{Construction}
Let $\mathcal{D}$ be a target distribution over the domain $\{0, 1, \dots, B-1\}$. The basic puzzle generation process, denoted $(k,h;g) \gets \mathsf{PUZ.Gen}(\ell)$, proceeds as follows: First, a secret round key $k \sim \mathcal{D}$ is sampled. Next, $k$ and the puzzle input label $\ell$ are strictly mapped to a challenge $c$ via a collision-resistant hash function. Finally, the PUF is physically evaluated at $c$, and the resulting response serves as the public helper data $h$ within the generated puzzle. Notice that generation requires only a single PUF invocation ($g=1$).

$$
\begin{array}{rl}
\multicolumn{2}{l}{\mathsf{PUZ.Gen}(\ell):} \\
1. & k \sim \mathcal{D} \\
2. & c \gets \mathrm{Hash}(k \parallel \ell) \\
3. & h \gets \mathsf{PUF}(c) \\
4. & \mathrm{Return}\ (k, h; 1)
\end{array}
$$

To subsequently solve the basic puzzle via $(\hat{k}; s) \gets \mathsf{PUZ.Solve}(\ell,h)$, a solver systematically enumerates candidate keys from the domain, starting sequentially from $0$. For each candidate $\hat{k}$, the solver derives the corresponding challenge and queries the PUF to obtain a physical response $\hat{r}$. The search terminates immediately upon identifying a key $\hat{k}$ whose response distance to the helper data $h$ falls within a predefined noise threshold $\tau \in (0, 1)$ (i.e., Hamming distance $d_H(\hat{r}, h) \le \tau \cdot m$).

$$
\begin{array}{rl}
\multicolumn{2}{l}{\mathsf{PUZ.Solve}(\ell, h):} \\
1. & \mathrm{For}\ \hat{k} = 0 \text{ to } B - 1 \ \mathrm{do:} \\
2. & \quad \hat{c} \gets \mathrm{Hash}(\hat{k} \parallel \ell) \\
3. & \quad \hat{r} \gets \mathsf{PUF}(\hat{c}) \\
4. & \quad \mathrm{If}\ d_H(\hat{r}, h) \le \tau \cdot m \ \mathrm{then:} \\
5. & \quad \quad \mathrm{Return}\ (\hat{k}; \hat{k}+1) \\
6. & \mathrm{Return}\ (\bot; B)
\end{array}
$$
If $\bot$ is returned, it is considered a failure in the search, which means that the valid $\hat{k} = k$ has been missed, potentially due to the noise in the evaluation of the PUF. 

\subsection{Optimal Solving Strategy}\label{sec:opt_solving}

\paragraph{Optimal Candidate Enumeration Strategy.}
An optimal solver should evaluate puzzle key candidates in order of strictly decreasing likelihood. This logically requires a structural property within the underlying key distribution that explicitly aligns the probability order with a sequential numerical search starting from $0$. We formally define this required property below and prove that the greedy enumeration strategy is strictly optimal.

\begin{definition}[Decreasing Distribution of $K$]\label{def:decreasingD}
Let $K$ be a random variable representing the secret round key. The distribution of $K$ is defined as strictly decreasing over its domain if, for all integers $i$ and $j$ such that $0 \le i < j < B$:
$$ \Pr[K = i] > \Pr[K = j] $$
\end{definition} 

Assuming the round keys are drawn from a decreasing distribution, a greedy approach, that tests the most probable candidates first, strictly minimizes the expected solving effort.

\begin{lemma}[Optimality of Greedy Enumeration]\label{lem:greedy_opt}
Let $\mathsf{PUZ}$ be a basic PUF-based time-lock puzzle where the secret key $K$ is supported over the domain $\{0, 1, \dots, B-1\}$. If $K$ follows a decreasing distribution, sequentially enumerating candidate keys from $0$ to $B-1$ in the first $B$ steps minimizes the expected solving effort among all possible sequential strategies.
\end{lemma}

\begin{proof}
Suppose a solver evaluates candidates in some arbitrary sequential order, say $(k_0, k_1, \dots, k_{B-1})$. The solver terminates successfully at exactly step $i$ if and only if the true key is $K = k_i$. Thus, the total expected solving effort $\mathbb{E}[S]$ is the inner product of the search step indices and their corresponding true key probabilities:
$$ \mathbb{E}[S] = \sum_{i=0}^{B-1} (i+1) \cdot \Pr[K = k_i] $$ 

By the Rearrangement Inequality, the sum of products of two sequences is minimized when the sequences are sorted in opposite monotonic directions. Since the sequence of step indices is inherently strictly increasing, minimizing the expectation $\mathbb{E}[S]$ requires the corresponding sequence of probabilities $(\Pr[K = k_0], \Pr[K = k_1], \dots, \Pr[K = k_{B-1}])$ to be strictly decreasing. Therefore, aligning the evaluation sequence $k_i$ exactly with the numerical sequential order $0, 1, \dots, B-1$ maps the highest probabilities to the smallest step indices. This perfectly satisfies the Rearrangement Inequality condition, achieving the global minimum expected solving effort among all possible sequential strategies.
\end{proof}

\paragraph{Optimal Decision Rule for Noisy Responses.}
For each enumerated candidate $\hat{k}$, the solver must statistically decide if the resulting noisy PUF response $\hat{r}$ matches the public helper data $h$. This decision is rigorously formulated as a binary hypothesis test.

\begin{lemma}[Optimality of Hamming Distance Thresholding]\label{lem:tau_optimal}
Assume a PUF noise model follows Def.~\ref{def:puf-indist} with independent bit-flip probability $e_p < 1/2$. Let $\hat{k}$ be a candidate key yielding response $\hat{r}$. Consider the hypotheses:
\begin{itemize}
    \item $H_0: \hat{k} = k$ (The candidate key is correct, $\hat{r}$ is a noisy evaluation of $r$)
    \item $H_1: \hat{k} \ne k$ (The candidate key is incorrect, $\hat{r}$ is uniformly random)
\end{itemize}
The optimal Neyman-Pearson likelihood-ratio test for distinguishing $H_0$ from $H_1$ is mathematically equivalent to a simple Hamming distance threshold test $d_H(h, \hat{r}) \le \tau \cdot m$ for some threshold parameter $\tau \in (0, 1)$.
\end{lemma}

\begin{proof}
Under $H_0$, $\hat{r}$ deviates from $h$ purely due to i.i.d. bit flips with probability $e_p$. Under $H_1$, $\hat{r}$ is drawn uniformly from $\{0, 1\}^m$. The likelihood ratio is:
$$ \Lambda(\hat{r}) = \frac{\Pr(\hat{r} | H_0)}{\Pr(\hat{r} | H_1)} = \frac{(\frac{e_p}{1-e_p})^{d_H(h, \hat{r})} \cdot (1 - e_p)^m}{2^{-m}} $$

By the Neyman-Pearson lemma, the optimal decision rule accepts $H_0$ if $\Lambda(\hat{r}) \ge \eta$ for a chosen target boundary $\eta$. Since $e_p < 1/2$, the ratio $\frac{e_p}{1-e_p} < 1$. Consequently, $\Lambda(\hat{r})$ is strictly monotonically decreasing with respect to $d_H(h, \hat{r})$. Therefore, evaluating $\Lambda(\hat{r}) \ge \eta$ is strictly equivalent to bounding the Hamming distance below a threshold $\tau \cdot m$.
\end{proof}

Combining Lemma~\ref{lem:greedy_opt} and Lemma~\ref{lem:tau_optimal}, we formally conclude that any optimal puzzle solver, whether legitimate or adversarial, must logically follow the exact same solving strategy: performing a greedy sequential enumeration of candidates $\hat{k}$ together with a Hamming distance threshold test.

\subsection{Failure Probability Analysis}\label{sec:failure}
The solving procedure $\mathsf{PUZ.Solve}$ may fail to recover the correct key $k$ in two distinct ways: False Negatives (FN) and False Positives (FP). The decoding failure probability rigorously quantifies this incorrectness, defining how likely the solver is to output an incorrect key or to fail to output any key at all. We bound these two specific failure scenarios separately by $\rho$ and $\rho_{\bot}$, respectively.

\begin{definition}[Decoding Failure Bound] \label{def:rho} 
Let $\mathsf{PUZ}$ be a PUF-based time-lock puzzle with corresponding induced random variables $(K,H,G,\hat{K},S)$.
We say $\mathsf{PUZ}$ has a decoding failure probability tightly bounded by $\rho$ and $\rho_{\bot}$ if, for every helper data $h$:
\begin{align*}
\Pr[\hat{K} \neq K, \hat{K}\neq \bot \mid H=h] &\leq \rho, \quad \text{and} \\
\Pr[\hat{K}= \bot \mid H=h] &\leq \rho_{\bot}.
\end{align*}
Equivalently, for all $\ell$, the total failure probability is:
\begin{equation*}
1 - \frac{\sum_k \Pr\Biggl[\begin{aligned}
  &(k,h;\cdot)\leftarrow \mathsf{PUZ.Gen}(\ell), \\
  &(k;\cdot)\leftarrow \mathsf{PUZ.Solve}(\ell, h)
\end{aligned}\Biggr]}{\Pr[(\cdot,h;\cdot)\leftarrow \mathsf{PUZ.Gen}(\ell)]}
\end{equation*}
\end{definition}

\paragraph{FN Probability of $\mathsf{PUZ}$.}
Let $e_p$ denote the independent per-bit flip probability of the noisy PUF oracle model, as formulated in Definition~\ref{def:puf-indist}.
For a fixed PUF challenge, the Hamming distance between two independent evaluations inherently follows the binomial distribution \( \text{Binom}(m, e_p) \). The FN probability of $\mathsf{PUZ}$, denoted as $\rho_{\mathrm{FN}}$, represents the probability that the \emph{correct} round key candidate is rejected. This occurs exclusively when the evaluation noise exceeds the threshold, i.e., $d_H(\hat{r}, h) > \tau \cdot m$.
\[
\rho_{\mathrm{FN}} = \Pr[d_H(\hat{r}, h) > \tau \cdot m] = \sum_{j =  \lfloor \tau\cdot m \rfloor +1}^{m} \binom{m}{j}{e_p}^j (1-{e_p})^{m-j},
\]
Using Hoeffding's inequality, for any $\tau > e_p$ we have 
\( \rho_{\mathrm{FN}} \leq \exp\!\left(-2m(\tau-e_p)^2\right) \), which decays exponentially in $m$.

\paragraph{FP Probability of $\mathsf{PUZ}$.}
The probability that a uniformly random vector $\hat{r}$ (generated by evaluating an incorrect key candidate) accidentally falls inside the acceptable threshold region around the target helper data $h$ is:
\begin{align*}
\rho_{\mathrm{FP}}^{(1)} = \Pr\left[ d_H(\hat{r}, h) \le \tau \cdot m \right] = \frac{1}{2^m}\sum_{j=0}^{\lfloor\tau \cdot m\rfloor}\binom{m}{j}.
\end{align*}
By Hoeffding's inequality, for any threshold $\tau < 1/2$, we have \( \rho_{\mathrm{FP}}^{(1)} \le \exp\!\left(-2m(1/2-\tau)^2\right)\), which also decays exponentially with respect to the response length $m$.
During the $\mathsf{PUZ.Solve}$ enumeration process, the solver may iteratively test up to $B-1$ incorrect round key candidates. An incorrect key could be falsely accepted either due to PUF noise falling within the threshold or due to a hash collision with probability $\rho_{\mathrm{coll}}$ when mapping candidates to challenges. Applying a standard union bound, the overall FP probability across all possible incorrect candidates is strictly bounded by:
\[
\begin{split}
\rho_{\mathrm{FP}} &\le (B-1) \cdot (\rho_{\mathrm{FP}}^{(1)}+\rho_{\mathrm{coll}}) \\
&\le (B-1) \bigl(\exp\!\bigl(-2m(1/2-\tau)^2\bigr)+\rho_{\mathrm{coll}}\bigr).
\end{split}
\]
For a basic puzzle $\mathsf{PUZ}$, $\rho = (B-1) (\exp\!\big(-2m(1/2-\tau)^2\big)+\rho_{\mathrm{coll}})$, and $\rho_\bot = \exp\!\left(-2m(\tau-e_p)^2\right)$.

\subsection{Solving Delay Analysis}\label{sec:solvetimeanalysis}

The distribution of the random variable $S$, which represents the solving delay of a puzzle, is defined below.  
\begin{definition}[Solving-Cost Distribution $S$]\label{def:S}
For any PPT algorithm $\mathcal{A}$ possessing oracle access to $\mathsf{PUF}$, we define $S[\mathcal{A}^{\mathsf{PUF}}(h)]$ as the induced solving-cost distribution produced by $\mathcal{A}$ when given the public helper data $h$. Formally:
\begin{multline*}
\Pr[s \gets S[\mathcal{A}^{\mathsf{PUF}}(h)] \mid H=h] \\
= \Pr\bigl[(\cdot; s)\gets \mathcal{A}^{\mathsf{PUF}}(h) \mid H=h\bigr].
\end{multline*}
\end{definition}

\paragraph{Solving-Time Windows.}
A window $[\alpha,\beta]$ statistically bounds the solving-time period during which a legitimate solver will confidently complete the puzzle with probability at least $1-\delta$. When this boundary threshold is strictly satisfied, the puzzle construction successfully guarantees a $([\alpha,\beta],\delta)$-Solving-Time Window property.

\begin{definition}[Solving Time Window]\label{def:window_general}
Let $\mathsf{PUZ}$ be a PUF-based puzzle with corresponding induced random variables $(K,H,G,\hat{K},S)$, and let $\mathcal{A}^{\mathsf{PUF}}$ be a PPT solving algorithm provided with public helper value $h$. We say $\mathcal{A}^\mathsf{PUF}$ has a $([\alpha,\beta],\delta)$ solving-time window conditional upon success if, for all $h$ where $\Pr[H=h]>0$:
\[
\Pr\big[s\in [\alpha , \beta], s \gets S[\mathcal{A}^{\mathsf{PUF}}(h)] \mid \hat{K}=K, H=h \big] \;\ge\; 1-\delta,
\]
where $\delta$ bounds the window's concentration error.
\end{definition}

\paragraph{The Distribution of Key $K$.}\label{sec:ki_dist}
The precise selection of the sampling distribution $\mathcal{D}$ for the key $K$ directly dictates how tightly the expected solving time remains concentrated around its mean. Thus, assigning a carefully tailored distribution for the puzzle keys is critical.

As established in Lemma~\ref{lem:greedy_opt}, $\mathcal{D}$ is required to be a strictly decreasing distribution to ensure the solver optimality. %
A natural choice capturing these requirements is the \emph{Geometric Distribution}, where each guess succeeds independently with probability $p$. The corresponding probability mass function is simply given by:
\begin{align}
    \Pr[K = k] = p(1-p)^k, \qquad \text{for } k \ge 0.
\end{align}
This distribution strictly satisfies the decreasing property (Definition~\ref{def:decreasingD}) because for any $k_1 < k_2$, $(1-p)^{k_1} > (1-p)^{k_2}$ (for $0 < p < 1$), meaning $p(1-p)^{k_1} > p(1-p)^{k_2}$.

However, an unbounded geometric distribution is not suitable for physical hardware implementation, as real-world puzzle instances strictly demand a finite bounded key domain $k \in \{0, 1, \ldots, B-1\}$. Therefore, we define $\mathcal{D}$ as a \emph{Truncated Geometric Distribution} with its support strictly cutting off at boundary $B$. To safely preserve the exponentially decreasing structure of the distribution while facilitating clean on-device implementation, we employ a standard rejection-sampling method during the puzzle generation routine: repeatedly sample independent keys $K \sim \mathrm{Geom}(p)$, returning the first output satisfying $K < B$. This straightforwardly produces the conditional probability distribution:
\begin{equation}\label{eq:trunc-geom}
\begin{split}
\tilde{\Pr}[K = k] &= \Pr[K = k \mid K < B] \\
&= \frac{p(1-p)^{k}}{1 - (1-p)^B}, \quad \text{for } 0 \le k \le B - 1.
\end{split}
\end{equation}
This truncated geometric distribution also satisfies the decreasing property of Definition~\ref{def:decreasingD}. Thus, it is fully compatible with the optimal greedy solving strategy. Moreover, its bounded support allows us to safely apply standard concentration inequalities (e.g., Hoeffding's inequality) to bound the total solving time $S$.

Suppose ${K}$ follows the truncated geometric distribution in \eqref{eq:trunc-geom}. Let $a := (1-p)^B$. The mean and variance of the truncated geometric distribution are explicitly given by:
\begin{align}
    \mathbb{E}[K] &=\frac{1-p}{p} - \frac{B a}{1-a} \label{eq:trunc-mean}  \\ 
    \text{Var}[K] &=\mathbb{E}[K^2] - (\mathbb{E}[K])^2 = \frac{1-p}{p^2} - \frac{B^2 a}{(1-a)^2} \label{eq:trunc-var}
\end{align}
The moments of the truncated distribution can be expressed as the moments of the unbounded geometric distribution minus an exact analytical correction term. Both the mean and the variance are strictly decreased by truncation, with the magnitude of reduction controlled by the tail mass parameter $a = (1-p)^B$. When $a \ll 1$, the right-hand correction terms become negligible, and the truncated distribution tightly approximates its unbounded counterpart.

\begin{theorem}[Solving-Time Window of $\mathsf{PUZ}$]\label{thm:basic-window-hoeffding}
Let $\mathsf{PUZ}$ be the basic PUF-based puzzle with a key drawn from the truncated geometric distribution~\eqref{eq:trunc-geom} with parameter $p \in (0,1)$ and support $\{0, \ldots, B-1\}$. Then, conditioned on correct key recovery ($\hat{K} = K$), $\mathsf{PUZ.Solve}$ has a $([\alpha, \beta], \delta)$-solving-time window (Definition~\ref{def:window_general}) for any $\delta \in (0,1)$, where
\begin{align}
\alpha &= \mathbb{E}[S] - \Delta(\delta) = (1+\mathbb{E}[K]) - (B-1)\sqrt{\frac{\ln(2/\delta)}{2}}, \label{eq:alpha-hoeffding} \\
\beta  &= \mathbb{E}[S] + \Delta(\delta) = (1+\mathbb{E}[K]) + (B-1)\sqrt{\frac{\ln(2/\delta)}{2}}. \label{eq:beta-hoeffding}
\end{align}
\end{theorem}

\begin{proof}
The solving cost $S = 1+K$ is bounded in $[1, B]$, with range $B-1$. Applying Hoeffding's inequality:
\[
\Pr\!\big[\,|S - \mathbb{E}[S]| > \Delta\,\big]
\le
2\exp\!\left(-\frac{2\Delta^2}{(B-1)^2}\right)
\le \delta,
\]
we get \begin{equation*}
\Delta(\delta) = (B-1)\sqrt{\frac{\ln(2/\delta)}{2}}. \qedhere
\end{equation*}
\end{proof}

\subsection{Efficiency Analysis}\label{sec:efficiency}

\begin{definition}[PUF-based Puzzle Efficiency]\label{def:puz_efficiency}
Let $\mathsf{PUZ}=(\mathsf{PUZ.Gen},\mathsf{PUZ.Solve})$ be a PUF-based time-lock puzzle scheme. 
Let the induced random variables $(K,H,G,\hat{K},S)$ \emph{correspond to} $\mathsf{PUZ}$.
The \emph{efficiency gap} of $\mathsf{PUZ}$, denoted by $\mu$, is formally defined as the ratio of expected solving effort to expected generation effort:
\[
\mu \;:=\; \frac{\mathbb{E}[\,S \mid \hat{K}=K \,]}{\mathbb{E}[\,G \,]}.
\]
\end{definition}
A large efficiency parameter $\mu$ indicates that recovering a hidden key requires disproportionately more physical PUF evaluations than initially generating the puzzle.

In our basic puzzle construction, the $\mathsf{PUZ.Gen}$ algorithm calls the PUF exactly once to evaluate $h \gets \mathsf{PUF}(c)$. The $\mathsf{PUZ.Solve}$ algorithm sequentially evaluates candidate keys starting from $\hat{k} = 0$. Conditioned on succeeding exactly at the true key $K$ (i.e., no false positives accepted early), the solver evaluates exactly $K + 1$ candidates, requiring $K + 1$ PUF queries. Thus, $\mathbb{E}[S \mid \hat{K} = K] = \mathbb{E}[K] + 1$. 

\begin{align}
\mu &= \mathbb{E}[K] + 1 = \frac{1}{p} - \frac{B(1-p)^B}{1 - (1-p)^B} 
\end{align}
This result shows that the efficiency gap scales inversely with $p$ (for sufficiently large $B$).%

\section{Puzzle Composition \texorpdfstring{$\mathsf{PUZ^n}$}{PUZn}}\label{sec:composition_puz}
\subsection{Construction}
To improve the concentration bounds of the solving delay and to reduce the failure rate due to the PUF noises, we propose sequentially composing multiple independent basic puzzles to formulate a larger, more robust composite structure. Let \(n\) denote the total number of puzzle rounds. By defining the key domain size as $B = 2^q$, each internal round key \(k_i\) natively maps to an element within the Galois field \(\mathbb{F}_{2^q}\). During execution, each independent puzzle round yields a potentially noisy candidate key $\hat{k}_i$ evaluated via the basic puzzle $\mathsf{PUZ}$. The final key is subsequently reconstructed by passing these candidates through a Reed-Solomon code to tolerate PUF errors, followed by a final cryptographic hash.

\paragraph{Composite Puzzle Generation $\mathsf{PUZ^n.Gen}$.} 
The generation algorithm concatenates $n$ independent basic puzzle instances. Every basic puzzle instance is strictly independent from one another because of the unique identifier $i$. This process generates a vector of round keys $\vect{k} = (k_1, \dots, k_n)$ along with their corresponding helper data responses $\vect{r} = (r_1, \dots, r_n)$. A Reed-Solomon encoder adds $2\omega$ parity symbols $\vect{p} = (p_1, \dots, p_{2{\omega}})$ over \(\mathbb{F}_{2^q}\) which allows the correction of up to \(\omega\) errors or \(2\omega\) erasures. The final secret key \(\mathsf{key}\) is derived by hashing the concatenated key vector.
$$
\begin{array}{rl}
\multicolumn{2}{l}{\mathsf{PUZ^n.Gen}(\ell):} \\
1. & \mathrm{For}\ i = 1 \text{ to } n \ \mathrm{do:} \\
2. & \quad (k_i, r_i; 1) \gets \mathsf{PUZ.Gen} (\ell \parallel i) \\
3. & \vect{k} \gets (k_1, \dots, k_n) \\
4. & (p_1, \dots, p_{2{\omega}}) \gets \mathsf{RS.Encode}(\vect{k}) \\
5. & \vect{h} \gets (r_1, \dots, r_n, p_1, \dots, p_{2{\omega}}) \\
6. & \mathsf{key} \gets \mathrm{Hash}(k_1 \parallel \dots \parallel k_n) \\
7. & \mathrm{Return}\ (\mathsf{key}, \vect{h} ; n)
\end{array}
$$

The output of $\mathsf{PUZ^n.Gen}(\ell)$ constitutes the required puzzle components. A complete physical time-lock puzzle can encrypt a solution message $msg$ using the derived key and compute a hash of the key $\mathsf{key}$ as $z \gets (\mathsf{Enc}_{\mathsf{key}}(msg), \mathsf{Hash}(\mathsf{key}))$. The final combined puzzle data released to the solver is $(z, \vect{h})$. $\mathsf{Hash}(\mathsf{key})$ functions as an integrity check validating the final recovered solution.

\paragraph{Composite Puzzle Solving $\mathsf{PUZ^n.Solve}$.}
Given the helper data $\vect{h}$, an integrated solver sequentially reconstructs the secret keys for each round. If a round solver fails to find a valid key candidate within the allowed Hamming distance $\tau \cdot m$, it returns the erasure symbol $\bot$. The recovered round keys, containing errors and erasures, are passed alongside the parity symbols to the Reed-Solomon decoder. Finally, the corrected keys are hashed, producing the final key $\mathsf{\hat{key}}$. 

$$
\begin{array}{rl}
\multicolumn{2}{l}{\mathsf{PUZ^n.Solve}(\ell, \vect{h}):} \\
1. & \mathrm{Parse}\ \vect{h} \text{ as } (r_1, \dots, r_n, p_1, \dots, p_{2{\omega}}) \\
2. & \mathrm{For}\ i = 1 \text{ to } n \ \mathrm{do:} \\
3. & \quad (\hat{k}_i; s_i) \gets \mathsf{PUZ.Solve}(\ell \parallel i, r_i) \\
4. & (\tilde{k}_1, \dots, \tilde{k}_n) \gets \mathsf{RS.Decode}(\hat{k}_1, \dots, \hat{k}_n, p_1, \dots, p_{2{\omega}}) \\
5. & \mathrm{If}\ \mathsf{RS.Decode}\ \text{fails, }\mathrm{Return}\ (\bot ; \sum_{i=1}^n s_i) \\
6. & \hat{\mathsf{key}} \gets \mathrm{Hash}(\tilde{k}_1 \parallel \dots \parallel \tilde{k}_n) \\
7. & \mathrm{Return}\ (\hat{\mathsf{key}} ; \sum_{i=1}^n s_i)
\end{array}
$$

If a final key $\hat{\mathsf{key}}$ is successfully returned, the solver parses $z$ and performs the integrity check $\mathsf{Hash}(\hat{\mathsf{key}}) \stackrel{?}{=} \mathsf{Hash}(\mathsf{key})$. Upon verification success, the solver transparently recovers the target message $msg = \mathsf{Dec}_{\hat{\mathsf{key}}}(\mathsf{Enc}_{\mathsf{key}}(msg))$. By modeling $\mathsf{Hash}$ as a random oracle, we assume any malicious solver using digital computation for additional help cannot bypass the basic puzzle solving by merely inverting $\mathsf{Hash}(\mathsf{key})$.

Unlike traditional software-based time-lock puzzles that tightly chain sequential outputs to delay solvers~\cite{abadi2021multi,abadi2025scalable}, our framework generates each basic puzzle completely independently. This intentional architectural choice isolates noise and minimizes error propagation between rounds, while strategically binding all puzzles sequentially to the strict hardware bottleneck of a single PUF chip. Furthermore, the usage of unique round identifiers effectively partitions the local search spaces so that resolving one puzzle round yields zero information for accelerating subsequent rounds.

\subsection{Optimal Solving Strategy for \texorpdfstring{$\mathsf{PUZ^n}$}{PUZn}}\label{sec:opt_solving_comp}
Lemma~\ref{lem:greedy_opt} establishes that, within a single isolated round, greedy sequential enumeration from $0$ to $B-1$ is uniquely optimal when the round key $K$ follows a decreasing distribution. In the composite puzzle $\mathsf{PUZ^n}$, however, the solver must simultaneously manage $n$ independent rounds, each requiring the successful recovery of a distinct key $k_i$. Crucially, due to the Reed-Solomon error-correction layer, not every individual round must be successfully solved for the overall decryption to succeed.

This introduces a natural strategic question: should the solver fully complete each round before advancing to the next, or interleave candidate evaluations across rounds to opportunistically reduce overall expected effort? The key complication arises from the inherent stochastic noise of physical PUF evaluations. Even when the true key is tested, the reconstructed response may fail the Hamming distance check, producing a false negative. Once the true key has been silently rejected, all further search effort within that round is wasted. A natural mitigation is to prematurely abandon partially-explored rounds and restart fresh ones, trading depth for breadth.

The theorem below formally characterizes the precise condition under which the simpler sequential greedy strategy, completing the current round exhaustively before advancing, is provably optimal, eliminating the need for such premature abandonment.

\begin{theorem}[Condition for Sequentially Greedy Optimality under Noise]\label{thm:noisy_greedy}
Let $\mathsf{PUZ^n}$ be a composite puzzle where each round key is drawn from the truncated geometric distribution (Eq.~\ref{eq:trunc-geom}) with parameter $p$ and cutoff $B$. Let $\rho_{\text{FN}}$ and $\rho_{\text{FP}}^{(1)}$ be the false negative and false positive probabilities per candidate evaluation, respectively. Suppose that the number of false negatives plus twice the number of false positives falls within the RS code decoding capacity. The sequentially greedy strategy (completing the current round before starting the next) strictly minimizes the expected total solving effort if and only if $\rho_{\text{FN}} < (1-p)^B (1-\rho_{\text{FP}}^{(1)})$.
\end{theorem}

\begin{proof}
Consider a solver that has sequentially tested and rejected candidates $0, 1, \dots, i-1$ within round $j$. Let $F_{<i}$ denote this event. For the solver to legitimately reach candidate $i$ without halting, it must have correctly rejected every earlier candidate. There are exactly two disjoint reasons this can occur:
\begin{enumerate}
    \item The true key $K_j \ge i$, and accordingly all $i$ earlier non-key candidates were correctly rejected. The probability of this scenario is $\Pr[K_j \ge i] \cdot (1-\rho_{\text{FP}}^{(1)})^i$.
    \item The true key $K_j < i$, but it was silently rejected by a false negative, while all $i-1$ remaining non-key candidates were also correctly rejected. The probability of this scenario is $\Pr[K_j < i] \cdot \rho_{\text{FN}} \cdot (1-\rho_{\text{FP}}^{(1)})^{i-1}$. 
\end{enumerate}

For notational convenience, define $\gamma = 1-\rho_{\text{FP}}^{(1)}$. The total probability of the rejection event $F_{<i}$ is therefore:
\[
\Pr[F_{<i}] = \frac{(1-p)^i - (1-p)^B}{1 - (1-p)^B} \gamma^i + \frac{1 - (1-p)^i}{1 - (1-p)^B} \rho_{\text{FN}} \gamma^{i-1}.
\]

The solver's optimal choice at depth $i$ is governed by the conditional probability of successfully finding the true key on exactly the \emph{next} candidate. This probability is:
\[
\Pr[K_j = i \cap \text{Accept } i] = \frac{p(1-p)^i}{1 - (1-p)^B} \gamma^i (1 - \rho_{\text{FN}}).
\]

The ratio of these two quantities is the conditional probability of success on step $i$:
\[
\begin{split}
f(i) &= \frac{\Pr[K_j = i \cap \text{Accept } i]}{\Pr[F_{<i}]} \\
&= \frac{p \gamma (1 - \rho_{\text{FN}}) (1-p)^i}{\bigl((1-p)^i - (1-p)^B\bigr)\gamma + \rho_{\text{FN}}\bigl(1 - (1-p)^i\bigr)}.
\end{split}
\]

To determine whether the solver should continue in the current round, we analyze whether the conditional success probability $f(i)$ strictly increases as depth $i$ grows. If $f(i)$ is strictly increasing with $i$, then the marginal probability of finding the key precisely at depth $i$ is higher than at any earlier depth $j < i$ (including $j=0$ for an entirely fresh round). Consequently, the solver minimizes expected effort by always greedily choosing the currently deepest round, ensuring it never abandons a partially searched round.

Substituting $x = (1-p)^i$, which strictly decreases as $i$ increases, we rewrite $f$ as:
\[
\begin{split}
    f(x) &= \frac{p \gamma (1 - \rho_{\text{FN}}) x}{\bigl(x - (1-p)^B\bigr)\gamma + \rho_{\text{FN}}(1 - x)} \\
&= \frac{p \gamma (1 - \rho_{\text{FN}}) x}{x(\gamma - \rho_{\text{FN}}) + \rho_{\text{FN}} - \gamma(1-p)^B}.
\end{split}
\]

The sign of the derivative $\frac{\partial f}{\partial x}$ is entirely determined by $\rho_{\text{FN}} - \gamma(1-p)^B$. Since $x$ decreases as $i$ increases, the success probability $f(i)$ strictly \emph{increases} with depth $i$ if and only if $f(x)$ strictly decreases with $x$, requiring $\frac{\partial f}{\partial x} < 0$. This holds if and only if
\[
\begin{split}
\rho_{\text{FN}} - \gamma(1-p)^B < 0 &\implies \gamma(1-p)^B > \rho_{\text{FN}} \\
&\implies (1-\rho_{\text{FP}}^{(1)})(1-p)^B > \rho_{\text{FN}}.
\end{split}
\]
When this condition is satisfied, $f(i)$ strictly increases with $i$, meaning the relative concentration of true-key probability within the remaining search space consistently outpaces the accumulated risk from earlier false negatives. Conditioning on the assumption that the RS code can decode the resulting false negatives, the solver does not need to retry a false negative puzzle round as part of the solving strategy. Thus, continuing deeper in the current round dominates all alternative strategies, making the sequential greedy approach provably optimal.
\end{proof}

This condition is trivially satisfied when no FPs or FNs occur during solving; thus, the optimality claim remains valid for the adversary in the security analysis of Section~\ref{sec:secproof}.

\subsection{Failure Analysis for \texorpdfstring{$\mathsf{PUZ^n}$}{PUZn}}\label{sec:failure_comp}
We analyze the decoding failure probability for the composition of $n$ independent rounds.
For each round $i$, we define a random variable $F_i\in\{0,1,2\}$ encoding the outcome as:
\[
F_i =
\begin{cases}
0, & \text{if } \hat{k}_i = k_i \quad\text{(correct)},\\
1, & \text{if } \hat{k}_i = \perp \quad\text{(erasure)},\\
2, & \text{if } \hat{k}_i \neq k_i,\ \hat{k}_i\neq \perp \quad\text{(error)}.
\end{cases}
\]
Let $f_E$ denote the number of rounds that produce an error, and $f_\perp$ the number of rounds that yield an erasure. For an RS code with $2\omega$ parity symbols, decoding succeeds if and only if Eq.~\eqref{eq:rs_condition} holds:
\begin{equation}
2f_E + f_\perp \le 2{\omega}.
\label{eq:rs_condition}
\end{equation}
From Definition~\ref{def:rho}, for every independent puzzle round, the per-round outcome probabilities satisfy:
\[
\begin{split}
&\Pr[F_i = 2] \le \rho, \qquad \Pr[F_i = 1] \le \rho_\perp, \\
&\text{implying} \quad \mathbb{E}[F_i] \le 2\rho + \rho_\perp.
\end{split}
\]
Since the $F_i$ are independent and bounded within $[0,2]$, we can use Hoeffding’s inequality to bound the overall decoding failure probability of the composite puzzle $\mathsf{PUZ}^n$:
\[
\Pr\!\left[\sum_{i=1}^n F_i > 2\omega\right]
\le
\exp\!\left(
-\frac{2\bigl(2\omega - n(2\rho+\rho_\perp)\bigr)^2}
{4n}
\right),
\]
provided that $2\omega > n(2\rho+\rho_\perp)$.

\subsection{Solving Delay Analysis for \texorpdfstring{$\mathsf{PUZ^n}$}{PUZn}}\label{sec:solvingConcentration}
We analyze the total solving delay incurred by an honest solver executing $\mathsf{PUZ^n.Solve}$, defined as the total number of PUF invocations performed across all $n$ rounds and denoted $T$.
Each round independently contributes a bounded number of PUF invocations. Specifically, a round may conclude in one of three ways:
(i) a correctly recovered key, incurring a cost of $1+k_i$ PUF invocations;
(ii) a false negative (FN), where no candidate passes the threshold check and the entire key space of $2^q$ values is exhausted; or
(iii) a false positive (FP), where an incorrect key candidate prematurely passes the check.
By the Reed-Solomon decoding guarantee, successful decoding requires at most $2\omega$ FN rounds or at most $\omega$ FP rounds.

\begin{theorem}[Solving-Delay Bound]
\label{thm:time-concentration}
Let $T$ be the total solving delay measured as the total number of PUF invocations performed by the (legitimate) solver across $n$ puzzle rounds, and let $\omega$ be the maximum number of error corrections by the Reed-Solomon encoder.
Then for any $\delta\in(0,1)$, the solver completes solving the puzzle with probability at least $1-\delta$ within the interval
\[
T \in [\alpha-\Delta(\delta),\; \beta+\Delta(\delta)],
\]
where
\begin{align*}
\alpha &= (n-\omega)(1+\mathbb{E}[K]) + \omega, \\
\beta  &= (n-2\omega)(1+\mathbb{E}[K]) + 2\omega\cdot 2^q,
\end{align*}
and
\[
\Delta(\delta)=2^q\sqrt{\frac{n}{2}\ln\!\left(\frac{2}{\delta}\right)}.
\]
\end{theorem}

\begin{proof}[Proof]
The total solving delay $T$ is the sum of $n$ bounded per-round contributions.
For rounds that are correctly decoded by the RS code, the solver performs $1+K_i \in [1, 2^q]$ PUF invocations,
where $K_i$ follows the truncated geometric distribution.
For FN error rounds, the solver enumerates the entire key space, resulting in exactly $2^q$ PUF invocations.
For FP error rounds, we conservatively assume that the solver performs only a single PUF invocation to establish the lower bound.

Conditioned on a successful RS decoding, we know that either at most $2\omega$ rounds with FN error or at most $\omega$ rounds with FP error can occur. Hence, the lower and upper bounds of solving delay window are obtained by considering the worst-case and
best-case distributions of FN and FP error rounds permitted by the RS decoding. The expected solving delay of the honest solver is bounded as
\[
\mathbb{E}[T] \in
\bigl[
(n-\omega)(1+\mathbb{E}[K]) + \omega \, \,\;,\;\;
(n-2\omega)(1+\mathbb{E}[K]) + 2\omega\cdot 2^q
\bigr],
\]
which defines the interval $[\alpha,\beta]$. Applying Hoeffding’s inequality to the sum of these bounded contributions yields the stated theorem.
\end{proof}
\noindent\paragraph{Realistic Solving Delay Window.}
The bound in Theorem~\ref{thm:time-concentration} is RS-aware and distribution-free, producing a conservative worst-case window. In practice, however, the dominant contribution to the honest solver's delay $T$ comes exclusively from rounds that are correctly decoded without FN or FP errors. In this typical regime, $T$ is well-approximated as a sum of $n$ i.i.d.\ contributions $1+K_i$, where each $K_i$ follows the truncated geometric distribution.

Since per-round delays have finite variance $\sigma^2 = \mathrm{Var}[K]$, the Central Limit Theorem guarantees convergence of the normalized sum to a Gaussian as $n$ grows. To replace this asymptotic approximation with a rigorous finite-sample guarantee, we apply the Berry-Esseen theorem. Specifically, we use Shevtsova's tightened bound~\cite{shevtsova2011absolute}, which bounds the maximum CDF approximation error between the normalized sum and the standard normal CDF $\Phi(z)$, provided the absolute third central moment $\varrho = \mathbb{E}[|K - \mathbb{E}[K]|^3]$ satisfies $\varrho \ge 1.286\sigma^3$:

\begin{multline*}
\sup_z \Biggl| \Pr\!\biggl( \frac{T - \mathbb{E}[T]}{\sqrt{\mathrm{Var}[T]}} \le z \biggr) - \Phi(z) \Biggr| \\
\le \frac{0.3328(\varrho + 0.429\sigma^3)}{\sigma^3 \sqrt{n}} := \Delta_{\text{BE}}.
\end{multline*}
where $\mathbb{E}[T]=n(1+\mathbb{E}[K])$ and $\mathrm{Var}[T]=n\sigma^2$. For practical puzzle compositions where $p \ll 1$, the geometric distribution closely approximates an exponential decay, yielding a skewness ratio $\frac{\varrho}{\sigma^3} \approx \frac{12}{e}-2 \approx 2.4145$, which satisfies Shevtsova's prerequisite. Substituting this value directly produces a tight closed-form bound on the maximum deviation:
\[
\Delta_{\text{BE}} \approx \frac{0.3328(2.4145 + 0.429)}{\sqrt{n}} \approx \frac{0.946}{\sqrt{n}}.
\]
As a concrete example, a composition of $n = 1000$ rounds guarantees a maximum CDF deviation of at most $\Delta_{\text{BE}} \approx 0.03$ (i.e., $3\%$).

To construct a strict $([\alpha,\beta],\delta)$ solving-time window, we absorb this approximation error into the normal tails by setting $z_{\text{BE}} = \Phi^{-1}(1 - \delta/2 + \Delta_{\text{BE}})$, which guarantees $1-\delta$ coverage and yields the tighter, realistic bounds:
\begin{align*}
    \alpha= \mathbb{E}[T] - z_{\text{BE}}\sqrt{\mathrm{Var}[T]}, 
    \qquad 
    \beta= \mathbb{E}[T] + z_{\text{BE}}\sqrt{\mathrm{Var}[T]}.
\end{align*}
This formulation exploits geometric concentration properties rather than worst-case limits, producing a substantially tighter analytical window that accurately reflects practical implementations, as experimentally verified in Section~\ref{sec:Experimets}.

\subsection{Efficiency Analysis for \texorpdfstring{$\mathsf{PUZ}^n$}{PUZn}}\label{sec:efficiency_comp}
Following Definition~\ref{def:puz_efficiency}, the efficiency $\mu_n$ evaluates the expected honest solving delay relative to the generation effort. Because $\mathsf{PUZ^n.Gen}$ constructs exactly $n$ independent basic puzzles, and its solving demands executing exactly $n$ sequential rounds with an expected total number of PUF evaluations of $n(1+\mathbb{E}[K])$, both efforts scale linearly by the exact same composition multiplier. Consequently, $\mu_n = \mu$. This shows that our composition successfully scales the absolute time delay and sharply restricts solver variation, without sacrificing the efficiency gap between the puzzle generator and solver.

\section{Security Analysis}\label{sec:secproof}
This section formally establishes the security of the composite puzzle $\mathsf{PUZ}^n$ in terms of early-solve hardness (Definition~\ref{def:early}). We prove security by analyzing a strictly more powerful adversary than the legitimate solver, and bounding their ability to bypass physical PUF queries using an information-theoretic analysis combined with Hoeffding's inequality.

\paragraph{Adversarial Model and Extra Capability.}\label{sec:adv_model}
Consider an adversary $\mathcal{A}$ attempting to solve the composite $n$-round puzzle $\mathsf{PUZ}^n$ significantly faster than the honest solving-delay window. The adversary runs at most $a$ classical computational steps and makes $s_{\mathcal{A}}$ total physical PUF invocations. To obtain a robust worst-case security bound, we grant the adversary the following idealized capabilities:
\begin{enumerate}
    \item \textbf{Noise-Free Oracle Access:} The adversary has perfect, noiseless oracle access to the PUF, incurring no false positive or false negative errors during solving.
    \item \textbf{Oracle Testing Budget:} The adversary has an oracle testing budget of $a$ steps that can finish within no time. This budget models some additional classical computational power that an adversary is willing to spend to speed up the computation. To cleanly capture the capability of this additional computation budget, we conservatively model each algorithmic step as eliminating at most one key candidate (e.g., by evaluating a guessed combination against the final $\mathsf{Hash}(\mathsf{key})$). This grants the adversary a maximum total information-resolving capacity of $\log_2 a$ bits.
\end{enumerate}

While the above idealized capabilities grant the adversary maximum algorithmic and informational leverage, we do not assume the attackers can launch invasive physical manipulations on the PUF. In practical deployments, our assumptions implicitly require a trusted, tamper-evident hardware boundary, which may necessitates additional active sensors. We note that because our P-TLP requires exact replication of the PUF responses for successful decoding, invasive physical attacks may lead to the degradation of the PUF reliability. Such degradation may ultimately cause decoding failures rather than enabling a faster, early-solve vulnerability.

The goal of the adversary is to leverage the additional oracle testing budget alongside the structural redundancies of the puzzle to recover the secret while minimizing their total number of physical PUF invocations $s_{\mathcal{A}}$. 

\begin{definition}[Early-Solve Hardness]\label{def:early}
Let $\mathsf{A}(a)$ denote the set of PPT algorithms $\mathcal{A}$ with oracle access to PUF $\mathsf{P}$, which on input helper data $h$ outputs a candidate key $\hat{k}$ in at most $a$ parallel classical computational steps with solving cost $s$ (the number of PUF invocations), i.e., $(\hat{k},s)\leftarrow \mathcal{A}^{\mathsf{P}}(h)$.
We say that the puzzle satisfies \emph{early-solve hardness at $\alpha$} if there exists a function $\alpha_{\mathrm{eff}}(a) = \alpha - O(\log a)$ and a negligible function $\mathrm{negl}(\cdot)$ such that for every $a$, every $\mathcal{A}\in\mathsf{A}(a)$, and every solving effort $\xi < \alpha_{\mathrm{eff}}(a)$,
\begin{multline*}
\Pr\bigl[s \le \xi \wedge \hat{k} = K \bigm| (\hat{k};s) \gets \mathcal{A}^{\mathsf{P}}(h), H=h\bigr] \\
   \;\le\; \mathrm{negl}\!\left(\alpha_{\mathrm{eff}}(a) - \xi \right).
\end{multline*}
\end{definition}
\noindent \textit{Note.}
The quantity $(\alpha_{\mathrm{eff}}(a) - \xi)$ is the time margin by which the adversary attempts to solve the puzzle ahead of the effective threshold $\alpha_{\mathrm{eff}}(a)$.

\paragraph{Information-Theoretic Constraints.}\label{sec:adv_constraints}
The $\mathsf{PUZ}^n$ is encapsulated using a RS code over $\mathrm{GF}(2^q)$, with $2\omega$ parity symbols appended to accommodate benign evaluation errors. From an adversarial perspective, these parity symbols represent exploitable informational redundancy that can be used to resolve puzzle uncertainty without additional PUF queries.

Since each puzzle symbol encodes $q$ bits, the $2\omega$ RS parity symbols allow the adversary to resolve exactly $2\omega \cdot q$ bits of uncertainty via erasure decoding without any PUF queries. Combined with the oracle testing budget, the total amount of uncertainty the adversary can resolve without physically querying the PUF is bounded by:
$E_{\text{max}} \le 2\omega \cdot q + \log_2 a$ bits.

To minimize physical evaluations, the adversary may selectively skip the search space in certain puzzle rounds. Let $g_i$ denote the number of candidate keys omitted from PUF evaluations in round $i$, which introduces $\log_2 g_i$ bits of residual uncertainty. Any successful solving strategy must satisfy:
\begin{equation}\label{eq:adv_entropy_limit}
\sum_{i=1}^n \log_2 g_i \le 2\omega \cdot q + \log_2 a.
\end{equation}

\paragraph{Optimal Adversarial Skipping Strategy.}\label{sec:ptlp-security-bound}
Given the entropy budget in Eq.~\eqref{eq:adv_entropy_limit}, the adversary seeks the allocation of unqueried candidates $\{g_i\}$ that maximizes the total number of skipped PUF invocations $\sum_i g_i$. This formulates as:
\[
\begin{aligned}
&\text{maximize} && \sum_{i=1}^{n} g_i \\
&\text{s.t.} && \sum_{i=1}^{n} \log_2 g_i \le E_{\text{max}}, \quad 0 \le g_i \le 2^q.
\end{aligned}
\]
The sum $\sum_i g_i$ is maximized when the budget is concentrated rather than spread evenly. Given the per-round cap $g_i \le 2^q$, the optimal strategy is to set as many rounds as possible to $g_i = 2^q$ (i.e., skip the round entirely).

Hence, the maximum number of completely skipped rounds is $2\omega + \left\lfloor \frac{\log_2 a}{q} \right\rfloor$. Note that while list decoding can theoretically correct errors beyond the threshold~\cite{brakensiek2023generic}, it is only useful when the attacker is querying a noisy PUF. Because our idealized security model assumes a best-case scenario for the adversary with noise-free PUF access, list decoding provides no additional advantage. Hence, the adversary must physically solve each remaining round by sequentially querying the PUF.

\begin{theorem}[Early-Solve Hardness of $\mathsf{PUZ}^n$]\label{thm:ptlp-security}
Let the effective early-solve threshold be $\alpha_{\mathrm{eff}}(a) = n'(1+\mathbb{E}[K])$, where $n' = n - 2\omega - \left\lfloor \frac{\log_2 a}{q} \right\rfloor$ is the number of rounds the adversary must physically solve. Then for every $a$, every $\mathcal{A}\in\mathsf{A}(a)$, and every candidate solving effort $\xi < \alpha_{\mathrm{eff}}(a)$, the puzzle $\mathsf{PUZ}^n$ satisfies:
\begin{equation}\label{eq:ptlp_early_solve}
\Pr\!\big[\, s_{\mathcal{A}} \le \xi \ \wedge\ \hat{\mathsf{key}}  = \mathsf{key} \,\big] \;\le\; \epsilon(\alpha_{\mathrm{eff}}(a) - \xi),
\end{equation}
where $\epsilon(x) := \exp\!\left(-\frac{2x^2}{n'(2^q-1)^2}\right)$ is negligible in $x$.
\end{theorem}

\begin{proof}
Since the adversary possesses noise-free oracle access, and by Lemma~\ref{lem:greedy_opt}, the greedy sequential strategy maximizes the cumulative probability of success at any bounded step limit, each of the remaining $n'$ required rounds costs exactly $1 + K_i$ PUF invocations, where $K_i \in [0, 2^q - 1]$. Therefore, the adversary's total physical solving cost is $s_{\mathcal{A}} = \sum_{i=1}^{n'} (1 + K_i)$.
The expected total cost exactly equals the effective threshold: $\mathbb{E}[s_{\mathcal{A}}] = n'(1+\mathbb{E}[K]) = \alpha_{\mathrm{eff}}(a)$. Since each summand $1 + K_i$ is bounded in $[1, 2^q]$ with range $2^q - 1$, applying Hoeffding's inequality directly yields that for any target $\xi < \alpha_{\mathrm{eff}}(a)$:
\begin{align*}\label{eq:hoeffding_adv}
\Pr[s_{\mathcal{A}} \le \xi] &= \Pr\!\left[s_{\mathcal{A}} - \alpha_{\mathrm{eff}}(a) \le -(\alpha_{\mathrm{eff}}(a) - \xi)\right] \nonumber \\
&\le \exp\!\left(-\frac{2(\alpha_{\mathrm{eff}}(a) - \xi)^2}{n'(2^q - 1)^2}\right).\qedhere
\end{align*}
\end{proof}

We assume honest solvers lack the additional computational resources necessary to perform advanced erasure-skipping strategies, and thus follow the standard sequential decoding. In practical deployments, the chosen parameters can explicitly target either the honest solving bound or the adversarial early-solve bound, depending on the assumed capabilities of the solver.

\section{Experiments}\label{sec:Experimets}
We physically evaluated our P-TLP design by implementing the composite puzzle $\mathsf{PUZ}^n$
on a Xilinx Artix-7 FPGA using the Digilent Nexys-A7 evaluation board, accessed from a host
machine via a UART interface.

\begin{figure}[t]
    \centering
    \includegraphics[width=\columnwidth]{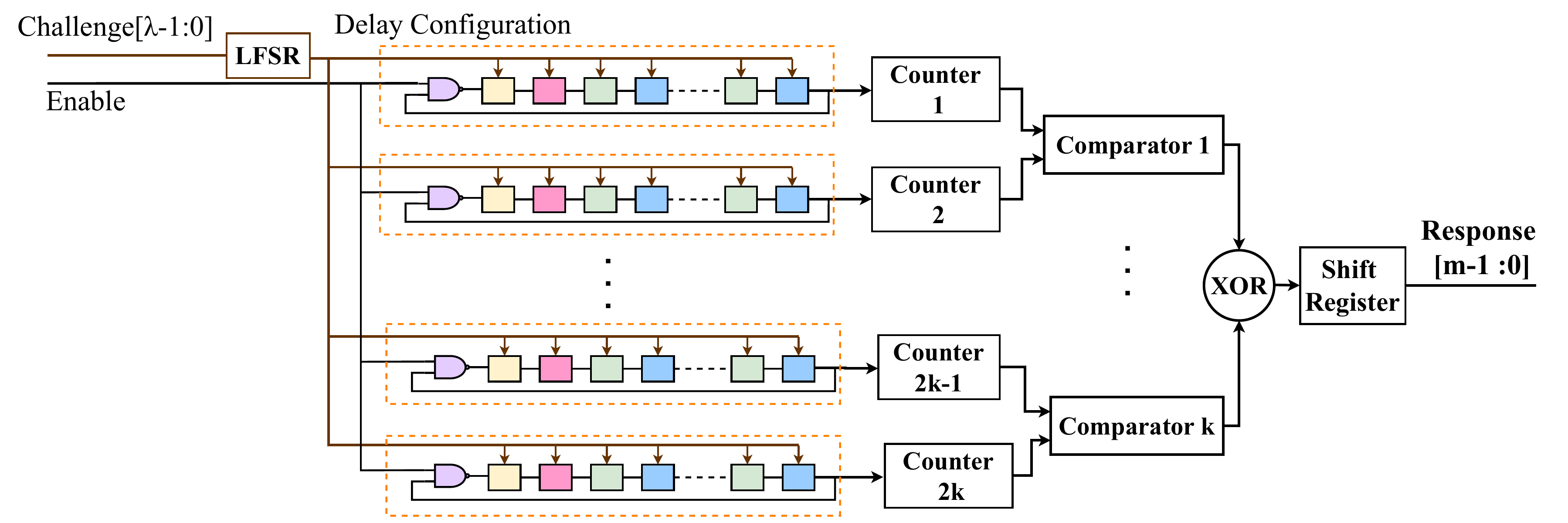}
    \caption{XOR Configurable Ring Oscillator PUF}
    \label{fig:CROPUF}
\end{figure}

\subsection{PUF Design and Hardware Implementation}\label{sec:PUF_imp}
The underlying physical primitive driving our P-TLP is an XOR Configurable Ring Oscillator (CRO) PUF implemented on the FPGA, illustrated in Figure~\ref{fig:CROPUF}. The architecture of each RO is inspired by the configurable construction detailed in~\cite{deng2020configurable}, utilizing hybrid delay-configurable units (DCU) to dynamically shift oscillator behavior. Note that our puzzle construction is generic, so it can be realized with any secure strong PUF. 

In our XOR CRO PUF, each RO consists of $64$ independent DCUs. Each DCU is composed of a pair of logic gates whose routing paths and effective propagation delays can be explicitly adjusted via a configuration input (CI) bit. Consequently, sweeping different CI configurations yields distinct intrinsic oscillation frequencies. Our implementation instantiates $10$ structurally identical configurable ROs arranged into $5$ disjoint pairs. For each pair, the oscillation frequencies are allowed to run freely and are measured by high-resolution counters over a fixed observation window before passing into digital comparators. The PUF output response for a given configuration is exactly the combined XOR of these $5$ parallel comparison results. Mathematically, each pair of the configurable ROs can be modeled in the same way as an arbiter PUF~\cite{Gassend2002}, so technically the entire XOR CRO PUF can be considered as a slower implementation of a 5-XOR arbiter PUF~\cite{Suh2007}. 

For protocol integration, the $\lambda$-bit puzzle challenge state directly seeds a linear feedback shift register (LFSR) where $\lambda=64$. The LFSR deterministically expands the challenge seed to extract distinct configuration vectors internally on-chip, dynamically feeding the arrays to output an accumulated $m$-bit response vector where $m=64$.

A defining advantage of the CRO PUF, unlike classical Arbiter PUFs, is the designer's ability to explicitly control the physical evaluation latency. The physical latency of a single query is strictly governed by the \texttt{count\_depth} parameter. This parameter dynamically locks the number of dedicated clock cycles over which the RO oscillations are aggregated, establishing a constant time $\tau_{\mathrm{P}}$ per PUF evaluation. 

In our deployments, we bound $\tau_P$ and balance noise profiles by fixing \texttt{count\_depth} to $10{,}000$. This specific depth suppressed environmental noise to deliver a robust intra-device physical reliability of approximately $98\%$, establishing a response uniformity clustered between $42\%$ and $44\%$. Crucially, because the trusted computing base is strictly limited to the FPGA encompassing the LFSR, the PUF, and the UART interface, any other puzzle solving operations occur on an untrusted host computer. We assume the physical TCB provides a trusted, tamper-proof hardware enclosure that prevents the adversary from modifying the internal logic to maliciously reduce the evaluation latency. Hence, the adversarial physical latency bounds analyzed in our proofs apply holistically to this entire hardware TCB. While an adversary remains unconstrained to implement and accelerate their software solving strategy, Theorem~\ref{thm:noisy_greedy} has already proven the optimal strategy for minimizing the expected solving time; thus, our rigid non-parallelizable hardware latency firmly bounds the solving delay.

\begin{table*}[t]
\caption{%
For each parameter setting, a random puzzle instance is generated, and then the solving algorithm is evaluated over the specified number of runs. %
}
\label{table_result1}
\centering
\resizebox{\textwidth}{!}{%
\begin{tabular}{|c c c c||c c||c c c c||c c|}
\hline
\multicolumn{4}{|c||}{\textbf{Parameters}} &
\multicolumn{2}{c||}{\textbf{Theoretical prediction}} &
\multicolumn{4}{c||}{\textbf{Experimental results}} &
\multicolumn{2}{c|}{\textbf{Errors} \text{(/ run)}} \\
\hline
$n$ & $B$ & $p$ & \#Runs &
$\mathbb{E}[K]$ &
$\mathbb{E}[T]$ &
$^\dagger$ Mean $[T]$ &
$^\ddagger$STD $[T]$ &
Gen. time (s) &
Solve. time (s) &
FN &
FP \\
\hline
100 & 1000 & 0.001 & 100 &
417 &
41,848 &
42,319 &
0 &
5.30 &
678.82 &
1 &
0 \\
\hline
200 & 1000 & 0.001 & 100 &
417 &
83,697 &
81,531 &
89.7 &
7.26 &
1,307.77 &
2 &
0.27 \\
\hline
300 & 5000 & 0.001 & 100 &
965 &
289,850 &
273,000 &
132.1 &
10.65 &
4,374.26 &
1 &
0.20 \\
\hline
\end{tabular}
}
\footnotesize{$^\dagger$ The mean of the measured solving delay $T$ (the total PUF invocations no.) over runs.\\
$^\ddagger$ The standard deviation of the measured solving delay $T$ (the total PUF invocations no.) over runs. }
\end{table*}

\begin{table*}[t]
\caption{%
For each parameter setting, independent puzzle instances are generated and solved. Theoretical values are computed from the truncated geometric model, while experimental results are averaged over the specified number of runs. %
}\label{table_result2}
\centering
\resizebox{\textwidth}{!}{%
\begin{tabular}{|c c c c||c c c||c c c c||c c|}
\hline
\multicolumn{4}{|c||}{\textbf{Parameters}} &
\multicolumn{3}{c||}{\textbf{Theoretical prediction}} &
\multicolumn{4}{c||}{\textbf{Experimental results}} &
\multicolumn{2}{c|}{\textbf{Errors} \text{(/ run)}} \\
\hline
$n$ & $B$ & $p$ & \#Runs &
$\mathbb{E}[K]$ &
$\mathbb{E}[T]$  &
$\sigma_T$  &
$^\dagger$Mean $[T]$  &
$^\ddagger$ STD  $[T]$&
Gen. time (s) & 
Solve. time (s) &
FN &
FP \\
\hline
100 & 1000 & 0.001 & 100 &
417 &
41,848 &
2,816 &
42,528 &
3,117 &
1.74 &
682.16&
0.98 &
0.06 \\
\hline
200 & 1000 & 0.001 & 100 &
417 &
83,697 &
3,983 &
84,353 &
3,821 &
3.34 &
1,352.93&
1.65 &
0.10 \\
\hline
1000 & 9000 & 0.0005 & 10 &
1,898&
1,899,010&
55,487&
1,925,477 &
59,341 &
21.80 &
8h, 34m, 24.41s&
7.8 &
2.4 \\
\hline
3000 & 10000 & 0.0005 & 10 &
1,931 &
5,796,746 &
99,738 &
5,938,065  &
103,273 &
50.32 &
26h, 24m, 14.46s &
20 &
6.6 \\
\hline
6000 & 30000 & 0.00015 & 1 &
6,329 &
37,978,607 &
453,104 &
38,549,120 &
/ &
125.97 &
7d+3h+36m+33.50s&
36 &
37 \\
\hline
\end{tabular}
}
\footnotesize{$^\dagger$ The mean of the measured solving delay $T$ (the total PUF invocations no.) over runs.\\
$^\ddagger$ The standard deviation of the measured solving delay $T$ (the total PUF invocations no.) over runs.}
\footnotesize{ }

\end{table*}

\begin{figure}[t]
  \centering
  \includegraphics[width=\columnwidth]{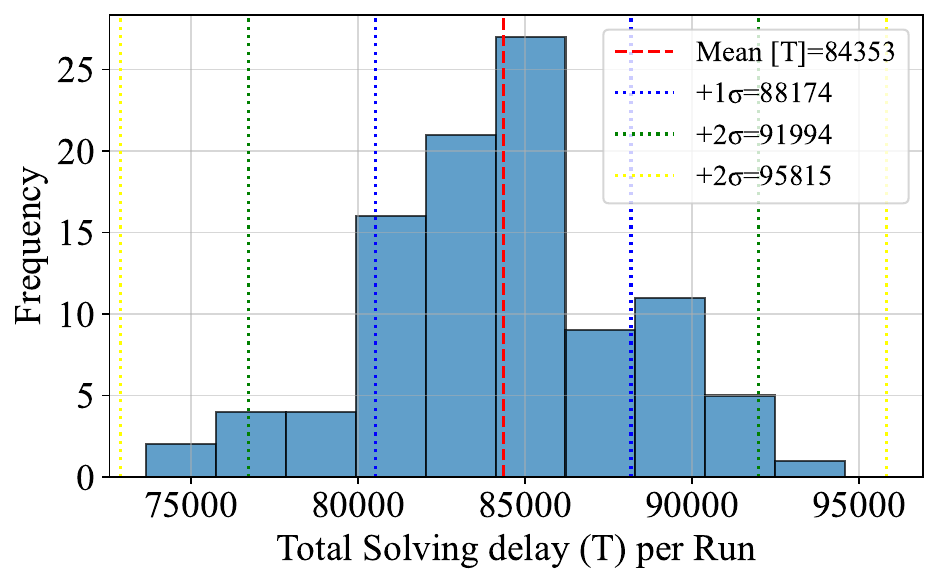}
  \caption{Solving time concentration of $\mathsf{PUZ}^n$ with $n=200$, $B=10^3$, $p=0.001$ in 100
  independent runs.}
  \label{fig:window_2}
\end{figure}

\subsection{Experimental Results Analysis}
We instantiated $\mathsf{PUZ}^n$ following Sections~\ref{sec:basicPUZ} and~\ref{sec:composition_puz}. We target various puzzle solving delays by adjusting $(n, p, B)$ while satisfying Theorem~\ref{thm:noisy_greedy}, and fix the Hamming distance threshold to $\tau=0.2$ based on our measured $\approx 2\%$ noise rate, balancing FN and FP error rates. Notably, early termination based on partial bit collisions is impossible, as the 64 configuration steps aggregate atomically in hardware and cannot be segmented externally.

Tables~\ref{table_result1} and~\ref{table_result2} detail our main experimental results, rigorously capturing hardware variance, evaluation latency, and intrinsic physical stability limits. To examine error stability versus algorithmic delay concentration independently, we split our methodology explicitly: Table~\ref{table_result1} observes solving delay concentration when running the optimal algorithm on a \emph{single, randomly generated puzzle instance} over 100 runs. Conversely, Table~\ref{table_result2} tests solving delay concentration given a fixed parameter setting by performing runs against \emph{independently generated random puzzle instances with the same parameter settings}.

\textbf{Decoding Stability and Intrinsic Noise Variance:}
The results documented in Table~\ref{table_result1} conclusively highlight profound stability within the physical evaluation process. Because the tested puzzle instance is fixed, the theoretical physical PUF evaluations $T$ should remain completely static, with practical execution variance resulting solely from FN/FP errors. Consequently, the measured standard deviation is demonstrably small compared to the puzzle solving time (e.g., $\sigma_T = 89.7$ spanning an execution mean of $\approx 81{,}531$ rounds under $n=200$). In the first row ($n=100$), physical evaluation drift generated exactly one single False Negative that repeated in the exact same round across the 100 runs, resulting in a recorded standard deviation of $0$. Crucially, the Reed-Solomon code could correct the small errors and successfully recover the secret across all experiments.

\textbf{Geometric Concentration Bounds vs Experimental Evaluation:}
Table~\ref{table_result2} evaluates our theoretically analyzed solving time concentration by solving random puzzle instances using identical parameter settings, confirming the robust delay predictions established in Section~\ref{sec:solvingConcentration}. In all cases involving more than 10 runs, the measured mean value of the solving time $T$ matches remarkably well with the theoretical prediction $\mathbb{E}[T]$. A clear trend indicates that the measured mean effectively converges to the theoretical prediction as the number of independent runs increases. Additionally, the standard deviations of evaluated distributions scaled exactly as expected (e.g., tracking empirically at $103{,}273$ against a theoretically projected value of $99{,}738$ iterations), verifying the correctness of our bounding analysis in Section~\ref{sec:solvingConcentration}. A plot of the solving time distribution for $n=200$ is illustrated in Figure~\ref{fig:window_2}. We further benchmarked a long 7-day puzzle instance at $n=6000$, successfully containing aggregate noise profiles to just 36 FN errors and 37 FP errors on average, which is beneath the Reed-Solomon decoding capacity.

\section{Potential Applications in Digital Legacy Management}\label{sec:applications}
Digital legacy management involves securing sensitive personal data, accounts, and virtual assets so that designated heirs can access them only after a prolonged delay following the owner's passing. Traditional software-based TLPs are poorly suited for this scenario: as purely digital artifacts, they can be freely copied and solved simultaneously by multiple parties on independent machines. Furthermore, since computational speed advances unpredictably over decades, setting an accurate delay for a puzzle that must survive until the far future is highly speculative, exposing the legacy to premature decryption. Our P-TLP resolves both issues by anchoring the solving process to a unique, unclonable physical hardware token. The puzzle cannot be duplicated or solved in parallel across machines, transforming it into a tangible physical ``will'' that only the designated heir in possession of the specific token can initiate and complete. This strict hardware binding inherently prevents remote attackers from copying or accelerating the puzzle, making P-TLP uniquely suited for secure digital legacy management.

\section{Future Work and Conclusion}
\label{sec:conclusion}
One practical consideration for long-term deployment is PUF aging: over extended periods, such as a decade, PUF aging increases inherent noise, potentially affecting puzzle reliability beyond what the Reed-Solomon code can accommodate. While decade-long stability falls outside our current scope and remains a practical limitation of the work, this degradation could theoretically be managed by chaining puzzles with gradually increased error-correction capability. Investigating the practical efficacy of these mitigations is an interesting direction for future work.

In conclusion, this paper introduced P-TLP, shifting the foundational anchor of delay from algorithmic steps to the intrinsic, hardware-bounded physical latency of silicon hardware. Our composite construction, supported by formal analysis and FPGA experiments, achieves a tightly concentrated and predictable solving-delay window, establishing a practical path toward hardware-anchored time-lock applications.

\section*{Acknowledgments} Chenglu Jin and Marten van Dijk are (partially) supported by project CiCS of the research programme Gravitation, which is (partly) financed by the Dutch Research Council (NWO) under the grant 024.006.037.

\bibliographystyle{IEEEtran}
\bibliography{bib/puf_detailed_20151112}

\end{document}